\documentclass{jpp}

\newcommand{\be}{\begin{equation}}
\newcommand{\ee}{\end{equation}}
\let\proof\relax 
\usepackage[psamsfonts]{amsfonts}
\usepackage{amsmath, amsthm, amssymb}
\usepackage{multirow}
\usepackage{enumerate}
\usepackage{setspace}
\usepackage[dvips]{graphicx}
\usepackage{natbib}

\newcommand{\derv}[1]{\frac{\partial}{\partial #1}}

\newcommand{\deriv}[2]{\frac{\partial #1}{\partial #2}}

\newcommand{\beqn}{\begin{equation}}
\newcommand{\eeqn}{\end{equation}}
\newcommand{\beqnar}{\begin{eqnarray}}
\newcommand{\eeqnar}{\end{eqnarray}}

\newtheorem{theorem}{Theorem}[section]
\newtheorem{proposition}[theorem]{Proposition}
\ifCUPmtlplainloaded \else
  \checkfont{eurm10}
  \iffontfound
    \IfFileExists{upmath.sty}
      {\typeout{^^JFound AMS Euler Roman fonts on the system,
                   using the 'upmath' package.^^J}%
       \usepackage{upmath}}
      {\typeout{^^JFound AMS Euler Roman fonts on the system, but you
                   dont seem to have the}%
       \typeout{'upmath' package installed. JPP.cls can take advantage
                 of these fonts, if you use 'upmath' package.^^J}%
      }
  \else
  \fi
\fi

\ifCUPmtlplainloaded \else
  \checkfont{msam10}
  \iffontfound
    \IfFileExists{amssymb.sty}
      {\typeout{^^JFound AMS Symbol fonts on the system, using the
                'amssymb' package.^^J}%
       \usepackage{amssymb}%
         \let\leq=\leqslant
         
      }{}
  \fi
\fi

\ifCUPmtlplainloaded \else
  \IfFileExists{amsbsy.sty}
    {\typeout{^^JFound the 'amsbsy' package on the system, using it.^^J}%
     \usepackage{amsbsy}}
    {\providecommand\boldsymbol[1]{\mbox{\boldmath $##1$}}}
\fi

\newsavebox{\astrutbox}
\sbox{\astrutbox}{\rule[-5pt]{0pt}{20pt}}

\title[Multi-Symplectic Magnetohydrodynamics: II]
{Multi-Symplectic Magnetohydrodynamics: II, Addendum and Erratum}

\author[G. M. Webb, J. F. McKenzie and G. P. Zank]%
{G.\ns M.\ns W\ls E\ls B\ls B$^1$%
 \thanks{Email address for correspondence: gmw0002@uah.edu},\ns
J.\ls F.\ns M\ls c\ls K\ls E\ls N\ls Z\ls I\ls E$^{1,3}%
 \thanks{deceased}$\break
 \and G.\ns P.\ns Z\ls A\ls N\ls K$^{1,2}$}

\affiliation{$^1$Center for Space Plasma and Aeronomic Research, 
The University of Alabama in Huntsville, 
Huntsville AL 35805, USA\\[\affilskip]
$^2$Department of Space Science, The University of
Alabama in Huntsville, Huntsville AL 35805, USA\\
$^3$Department of Mathematics and Statistics, 
Durban University of Technology,\\
Steve Biko Campus, Durban South Africa, and School of Mathematical Sciences,
University of KwaZulu-Natal, Durban South Africa}

\begin{document}


\maketitle

\begin{abstract}
A recent paper by \cite{Webb14c} on multi-symplectic magnetohydrodynamics 
(MHD) using  Clebsch variables in an Eulerian action principle
with constraints is further extended. We relate a class of 
symplecticity conservation laws to a vorticity conservation law, 
and provide a corrected form of the Poincar\'e-Cartan differential 
form formulation of the system. We also correct some typographical errors
(omissions) in \cite{Webb14c}. We show that the vorticity-symplecticity 
conservation law, that arises as a compatibility condition on the 
system, expressed in terms of the Clebsch variables is equivalent 
to taking the curl of the conservation form of the MHD momentum equation. 
We use the Cartan-Poincar\'e form to obtain a class of differential forms that represent the system 
using Cartan's geometric theory of partial differential equations.
\end{abstract}


\maketitle


\section{Introduction}
\index{symplectic}%
Multi-symplectic equations for
\index{Hamiltonian}%
Hamiltonian systems with two or more
independent variables $x^\alpha$ have been developed 
as a useful extension of
Hamiltonian systems with one evolution variable $t$.
This development has connections with dual
variational formulations
of traveling wave problems (e.g. \cite{Bridges92}), and is useful in 
numerical schemes for
Hamiltonian systems. Bridges and co-workers used the multi-symplectic
approach to study
 linear and nonlinear wave propagation, 
 generalizations of 
wave action, wave modulation 
theory, and wave stability problems (\cite{Bridges97a}, \cite{Bridges97b}). 
\cite{BridgesReich06} develop multi-symplectic difference schemes. 
Multi-symplectic Hamiltonian systems 
have been studied by 
\cite{Marsden99} and  \cite{Bridges05}.  
\cite{Webb07b,Webb08,Webb14d} discuss 
traveling waves in multi-fluid plasmas using a multi-symplectic formulation. 
 \cite{Holm98} 
give an overview of Hamiltonian systems, semi-direct 
product Lie algebras and Euler-Poincar\'e equations.  

\cite{Cotter07} developed a multi-symplectic, Euler-Poincar\'e
formulation of fluid mechanics.  They showed  
that  multi-symplectic ideal fluid mechanics type systems are 
 related to  Clebsch variable formulations  
 in which the Lagrange multipliers 
play the role of canonically conjugate momenta to the 
constrained variables.  
Thus, the Clebsch variable formulation involves a momentum map.  

The main aim of the present paper is to correct and extend the 
analysis of the multi-symplectic MHD equations derived by \cite{Webb14c}.
Section 2 describes the MHD equations and the first law of thermodynamics. 
An overview of the key equations of the multi-symplectic 
MHD equations obtained by \cite{Webb14c} is given in Section 3.  
We correct some typographical 
errors and present a more consistent description of the Cartan-Poincar\'e 
form for the system (section 4).  The Cartan-Poincar\'e
form is used to obtain a set of differential 
forms representing the system using Cartan's geometric theory of 
partial differential equations. 

Section 5 concludes with a summary and discussion.

\section{The Model}

The magnetohydrodynamic equations are:

\beqnar
&&\deriv{\rho}{t}+\nabla{\bf\cdot} (\rho{\bf u})=0, \label{eq:2.1}\\
&&\derv{t}(\rho{\bf u})+\nabla{\bf\cdot}\left[\rho {\bf u}{\bf u}
+\left(p+\frac{B^2}{2\mu}\right){\bf I} -\frac{\bf B B}{\mu}\right]=0,
\label{eq:2.2}\\
&&\deriv{S}{t}+{\bf u}{\bf\cdot}\nabla S=0, 
\label{eq:2.3}\\
&&\deriv{\bf B}{t}-\nabla\times\left({\bf u}\times {\bf B}\right)
+{\bf u}\nabla{\bf\cdot}{\bf B}=0, \label{eq:2.4}
\eeqnar
where $\rho$, ${\bf u}$, $p$,
$S$ and ${\bf B}$ are the gas density, fluid velocity, pressure,
specific entropy, and magnetic induction ${\bf B}$ respectively, and
${\bf I}$ is the unit $3\times 3$ dyadic.
The gas pressure $p=p(\rho,S)$ is a function of the density $\rho$ and
entropy $S$, and $\mu$ is the magnetic permeability. 
Equations (\ref{eq:2.1})-(\ref{eq:2.2}) are the mass and  momentum  
conservation laws, (\ref{eq:2.3}) is the entropy advection equation 
and (\ref{eq:2.4}) is Faraday's equation in the MHD limit. 
In classical MHD, (\ref{eq:2.1})-(\ref{eq:2.4}) are supplemented by
Gauss' law:
\beqn
\nabla{\bf\cdot}{\bf B}=0. \label{eq:2.5}
\eeqn
which implies the non-existence of magnetic monopoles. The MHD equations are closed 
by specifying an equation of state for the gas $U=U(\rho,S)$ where 
$U(\rho,S)$ is the internal energy per unit mass. The first law of thermodynamics
has the form:
\beqn
TdS=dQ=dU+pd\tau\quad\hbox{where}\quad \tau=\frac{1}{\rho}, \label{eq:2.7}
\eeqn
where  $\tau=1/\rho$ is the specific
volume. Using the internal energy per unit volume $\varepsilon=\rho U$
instead of $U$, (\ref{eq:2.7}) may be written as:
\beqn
TdS=\frac{1}{\rho}\left(d\varepsilon-hd\rho\right)\quad\hbox{where}\quad
h=\frac{\varepsilon+p}{\rho}, \label{eq:2.8}
\eeqn
is the enthalpy of the gas. Equation
 (\ref{eq:2.8}) gives  the formulae:
\beqn
 \rho T=\varepsilon_S, \quad h=\varepsilon_\rho,
\quad p=\rho\varepsilon_{\rho}-\varepsilon, \label{eq:2.9}
\eeqn
relating the temperature $T$, enthalpy $h$ and pressure $p$ to the
internal energy density $\varepsilon(\rho,S)$. 

\section{Multi-Symplectic Appproach}
In this section we give a brief overview of the multi-symplectic formulation of MHD by 
\cite{Webb14c}. 

\subsection{Clebsch Variable Action Principle}
Consider the MHD action (modified by constraints):
\beqn
J=\int\ d^3x\ dt  L,  \label{eq:Clebsch1}
\eeqn
where
\begin{align}
L=&\left\{\frac{1}{2}\rho u^2-\epsilon(\rho, S)-\frac{B^2}{2\mu_0}\right\}
+\phi\left(\deriv{\rho}{t}+\nabla{\bf\cdot}(\rho {\bf u})\right)\nonumber\\
&+\beta\left(\deriv{S}{t}+{\bf u}{\bf\cdot}\nabla S\right)
+\lambda\left(\deriv{\mu}{t}+{\bf u\cdot}\nabla\mu\right) \nonumber\\
&+\boldsymbol{\Gamma}{\bf\cdot}\left(\deriv{\bf B}{t}-\nabla\times({\bf u}\times{\bf B})
+{\bf u}(\nabla{\bf\cdot B})\right). \label{eq:Clebsch2}
\end{align}
In (\ref{eq:Clebsch2}) the Lagrange multipliers $\phi$, $\beta$, $\lambda$ and 
$\boldsymbol{\Gamma}$ ensure that the mass, entropy, Lin constraint and Faraday's equation 
are satisfied. Note we do not set $\nabla{\bf\cdot}{\bf B}=0$ in our analysis. 
It can be set equal to zero after all variational calculations for the MHD system
are finished. However, it does modify the Lagrange multiplier equation for $\boldsymbol{\Gamma}$
that ensures Faraday's equation is satisfied (see \cite{Morrison80}, \cite{Morrison82},\cite{Morrison82a}, 
\cite{Holm83a},\cite{Holm83b} and \cite{Chandre13} for further discussion 
of the MHD Poisson bracket). 
The Lagrangian in curly brackets equals the kinetic minus
the potential energy (internal thermodynamic energy plus magnetic energy).
The Lagrange multipliers $\phi$, $\beta$, $\lambda$, 
and $\boldsymbol{\Gamma}$ ensure that the 
mass, entropy, Lin constraint, Faraday equations are satisfied.

 Stationary point conditions for the action are $\delta J=0$.
 $\delta J/\delta {\bf u}=0$ gives the Clebsch representation
for ${\bf M}=\rho {\bf u}$:
\begin{equation}
{\bf M}=\rho {\bf u}=\rho \nabla\phi-\beta\nabla S-\lambda\nabla\mu+
{\bf B}\times(\nabla\times\boldsymbol{\Gamma})-\boldsymbol{\Gamma}\nabla{\bf\cdot}{\bf B}.
\label{eq:Clebsch3}
\end{equation}
 Setting $\delta J/\delta\phi$, $\delta J/\delta\beta$,
$\delta J/\delta \lambda$, $\delta J/\delta\boldsymbol{\Gamma}$  
equal to zero gives the mass, entropy advection, Lin constraint, 
and Faraday (magnetic flux conservation) constraint
equations:
\begin{align}
&\rho_t+\nabla{\bf\cdot}(\rho {\bf u})=0,\nonumber\\
&S_t+{\bf u}{\bf\cdot}\nabla S=0,\nonumber\\
&\mu_t+{\bf u\cdot}\nabla\mu=0, \nonumber\\
&{\bf B}_t-\nabla\times({\bf u}\times{\bf B})+{\bf u}(\nabla{\bf\cdot B})=0.
\label{eq:Clebsch5}
\end{align}

Similarly, setting $\delta J/\delta\rho$, $\delta J/\delta S$, $\delta J/\delta\mu$,
$\delta J/\delta {\bf B}$ equal to zero gives evolution equations 
for the Clebsch potentials $\phi$, $\beta$, $\lambda$ and $\boldsymbol{\Gamma}$
as:
\begin{align}
&-\left(\deriv{\phi}{t}+{\bf u}{\bf\cdot}\nabla\phi\right)
+\frac{1}{2} u^2-h=0, \label{eq:Clebsch6}\\
&\deriv{\beta}{t}+\nabla{\bf\cdot}(\beta {\bf u})+\rho T=0, 
\label{eq:Clebsch7}\\
&\deriv{\lambda}{t}+\nabla{\bf\cdot}(\lambda {\bf u})=0, 
\label{eq:Clebsch8}\\
&\deriv{\boldsymbol{\Gamma}}{t}
-{\bf u}\times (\nabla\times\boldsymbol{\Gamma})
+\nabla(\boldsymbol{\Gamma}{\bf\cdot u})+\frac{\bf B}{\mu_0}=0. 
\label{eq:Clebsch9}
\end{align}
Equation (\ref{eq:Clebsch6}) is Bernoulli's equation. 
The $\nabla(\boldsymbol{\Gamma}{\bf\cdot u})$ term 
in (\ref{eq:Clebsch9}) is associated with
$\nabla{\bf\cdot}{\bf B}\neq 0$.

\subsection{Multi-symplectic approach}
From \cite{Webb14c}, the MHD system was written in the multi-symplectic form:
\begin{equation}
{\sf K}^\alpha_{ij}\deriv{z^j}{x^\alpha}=\deriv{H}{z^i}, \label{eq:msymp1}
\end{equation}
where the variables $z^j$ are defined as:
\begin{equation}
{\bf z}=\left(u^1,u^2,u^3,\rho,S,\mu,B^1,B^2,B^3,\Gamma^1,\Gamma^2,\Gamma^3,\lambda,\beta,\phi\right)^T. 
\label{eq:msymp2}
\end{equation}
The variables: $x^\alpha=(t,x,y,z)$ are the independent space
and time variables, i.e. $x^0=t$, $x^1=x$, $x^2=y$, $x^3=z$. Below we use the notation $z^s_{,\alpha}=\partial z^s/\partial x^\alpha$. 
The constrained Lagrangian $L$ in (\ref{eq:Clebsch2}) is written in the form:
\begin{equation}
L=L_0+\sum_s L^\alpha_s z^s_{,\alpha}, \label{eq:msymp2a}
\end{equation}
where
\begin{equation}
L_0=\frac{1}{2}\rho u^2-\varepsilon(\rho,S)-\frac{B^2}{2\mu}, \label{eq:msymp2b}
\end{equation}
is the unconstrained Lagrangian.  

 The multi-symplectic Hamiltonian is given by the generalized Legendre transformation:
\begin{equation}
H(z)=\sum_s L^\alpha_s z^s_{,\alpha}-L=-L_0\equiv 
-\left(\frac{1}{2}\rho u^2-\varepsilon(\rho,S)-\frac{B^2}{2\mu}\right). \label{eq:msymp2c}
\end{equation}

The fundamental one forms $\omega^\alpha$
($\alpha=0,1,2,3$)
of the multi-symplectic system (\ref{eq:msymp1}) are defined as:
\begin{equation}
\omega^\alpha=L^\alpha_j(z)dz^j. \label{eq:msymp5}
\end{equation}
Taking the exterior derivative of the $\omega^\alpha$ gives the formulae
\begin{equation}
\kappa^\alpha=d\omega^\alpha=\frac{1}{2} {\sf K}^\alpha_{ij} dz^i\wedge dz^j, \label{eq:msymp5a}
\end{equation}
which, in turn, gives the formulae:
\begin{equation}
{\sf K}^\alpha_{ij}=\left(\deriv{L^\alpha_j}{z^i}-\deriv{L^\alpha_i}{z^j}\right), \label{eq:msymp4}
\end{equation}
for the skew symmetric matrices in (\ref{eq:msymp1}). 

For the variables $z^j$ in (\ref{eq:msymp2})
the one forms $\omega^\alpha$ are given by:
\begin{align}
\omega^0=&\phi d\rho+\beta dS+\lambda d\mu+\boldsymbol{\Gamma}{\bf\cdot}d{\bf B}, \nonumber\\
\omega^i=&u^i\left(\beta dS+\lambda d\mu+\phi d\rho\right)
+\rho \phi d u^i+\boldsymbol{\Gamma}{\bf\cdot B} du^i
-B^i(\boldsymbol{\Gamma}{\bf\cdot}d{\bf u})+u^i (\boldsymbol{\Gamma}{\bf\cdot}d{\bf B}), \label{eq:msymp7}
\end{align}
where $1\leq i\leq 3$. 

The pullback conservation laws:
\begin{equation}
G_\beta=D_\alpha\left(L^\alpha_j(z) z^j_{,\beta}-L\delta^\alpha_\beta\right)=0, \label{eq:msymp8}
\end {equation}
follow from the pullback of the identities:
\begin{align}
&\left(L^\alpha_j dz^j\right)_{,\alpha}=d\left\{L^\alpha_j(z)z^j_{,\alpha}-H(z)\right\}=dL, \nonumber\\
&\left(L^\alpha_j dz^j\right)_{,\alpha}=\left(L^\alpha_j z^j_{,\beta}dx^\beta\right)_{,\alpha}
=\left(L^\alpha_j z^j_{,\beta}\right)_{,\alpha}dx^\beta=dL=\deriv{L^\alpha}{x^\beta} dx^\beta. 
\label{eq:msymp9}
\end{align}
(see proposition 4.1 of \cite{Webb14c} for more detail).  

The pullback 
equation $\kappa^\alpha_{,\alpha}=0$ where $\kappa^\alpha=d\omega^\alpha$ gives rise to the symplecticity 
or phase-space conservation laws (structural conservation laws):
\begin{equation}
D_\alpha\left({\sf K}^\alpha_{ij} z^i_{,\beta}z^j_{,\gamma}\right)=0, \quad \beta<\gamma. \label{eq:msymp10}
\end{equation}
These conservation laws can also be written in the form:
\begin{equation}
D_\beta G_\gamma-D_\gamma G_\beta=D_\alpha\left({\sf K}^\alpha_{ij} z^i_{,\beta}z^j_{,\gamma}\right)=0, 
\label{eq:msymp11}
\end{equation}
i.e., the symplecticity conservation laws (\ref{eq:msymp11}) are compatibility conditions for the pullback 
conservation laws (\ref{eq:msymp8}).

The pullback conservation law (\ref{eq:msymp8}) for $\beta=0$ (\cite{Webb14c}, equation (5.39))
reduces to the energy conservation law:
\begin{equation}
G_0=-\left\{\derv{t}\left(\frac{1}{2} \rho u^2+\varepsilon (\rho,S)+\frac{B^2}{2\mu}\right) 
+\nabla{\bf\cdot}\left(\rho {\bf u}\left(\frac{1}{2} u^2+h\right)
+\frac{{\bf E}\times {\bf B}}{\mu}\right)
\right\}=0, \label{eq:msymp12}
\end{equation}
where $h=(\varepsilon+p)/\rho$ is the enthalpy of the gas, ${\bf E}=-{\bf u}\times{\bf B}$ 
is the electric field, 
and ${\bf E}\times{\bf B}/\mu$ is the Poynting flux. Similarly, the pullback conservation laws 
(\ref{eq:msymp8}) for $\beta=i$ ($1\leq i\leq 3$) give rise to the MHD momentum conservation equation:
\begin{equation}
G^i=-\left\{\derv{t}\left(\rho {\bf u}\right)+\nabla{\bf\cdot}
\left[\rho {\bf u}\otimes {\bf u}+
\left(p+\frac{B^2}{2\mu}\right){\bf I}-\frac{{\bf B}\otimes {\bf B}}{\mu}\right]\right\}^i=0, 
\label{eq:msymp13}
\end{equation}
(\cite{Webb14c}, equation (5.41)). 

\section{Extensions, comments and corrections}
The symplecticity conservation laws (\ref{eq:msymp10})-(\ref{eq:msymp11}) 
have a generalized curl form. Consider the symplecticity laws (\ref{eq:msymp11}) for $1\leq i,k\leq 3$, 
namely:
\begin{equation}
\Omega_{ik}=D_iG_k-D_kG_i=0. \label{eq:4.1}
\end{equation}
Introduce the dual of the tensor $\Omega_{ik}$ defined as:
\begin{equation}
V^p=-\frac{1}{2}\varepsilon_{pik}\Omega_{ik}=-(\nabla\times{\bf G})^p, \label{eq:4.2}
\end{equation}
where $\nabla\times{\bf G}$ is the spatial curl of ${\bf G}$. Taking into account the momentum conservation 
law (\ref{eq:msymp13}) for ${\bf G}$, (\ref{eq:4.2}) reduces to:
\begin{equation}
\nabla\times{\bf G}=-\left\{\derv{t}\nabla\times{\bf M}
+\nabla\times\left[\nabla{\bf\cdot}\left({\bf M}\otimes {\bf u}
-\frac{{\bf B}\otimes{\bf B}}{\mu}\right)\right]\right\}=0, \label{eq:4.3}
\end{equation}
where
\begin{equation}
{\bf M}=\rho {\bf u}, \label{eq:4.4}
\end{equation}
is the momentum density or mass flux ${\bf M}$ of the MHD fluid. Note there is no contribution 
from the magnetic pressure ($B^2/(2\mu)$) 
and gas pressure ($p$) gradient force terms in (\ref{eq:4.3}) because 
$\nabla\times\nabla(p+B^2/(2\mu))=0$ when one takes the curl of the momentum equation (\ref{eq:msymp13}). 
The evolution of $\nabla\times{\bf M}$ in (\ref{eq:4.3}) is thus determined by the inertia and 
magnetic tension components:
\begin{equation}
{\bf M}\otimes{\bf u}-\frac{{\bf B}\otimes{\bf B}}{\mu}, \label{eq:4.5}
\end{equation}
of the MHD stress-energy tensor. This suggests that (\ref{eq:4.3}) describes Alfv\'enic type disturbances, 
in which both fluid spin and magnetic tension forces are part of the dynamics. Equation (\ref{eq:4.3}) 
can also be expressed in the conservation law form:
\begin{equation}
\derv{t}\nabla\times{\bf M}
+\nabla_s\left[\nabla\times\left({\bf M}u^s-\frac{{\bf B}B^s}{\mu}\right)\right]=0.
\label{eq:4.6}
\end{equation}
Pressure gradient forces play no role in the vorticity-like conservation laws (\ref{eq:4.3}) 
and (\ref{eq:4.6}).

The symplecticity conservation law (\ref{eq:4.3}) is different than that obtained by taking the curl 
of Euler momentum equation in the form $d{\bf u}/dt={\bf F}$ where ${\bf F}$ is the net force on the fluid 
element, to obtain an equation for the evolution of the fluid vorticity 
$\boldsymbol{\omega}=\nabla\times{\bf u}$. \cite{Webb14a},\cite{Webb14b} and 
\cite{WebbAnco15}  
obtained a conservation law in ideal fluid mechanics (i.e. for ${\bf B}=0$) for the generalized vorticity 
\begin{equation}
\boldsymbol{\Omega}=\boldsymbol{\omega}+\nabla r\times\nabla S\quad\hbox{where}\quad \boldsymbol{\omega}=\nabla\times{\bf u}, \label{eq:4.7}
\end{equation}
is the fluid vorticity and $r$ satisfies the equation: 
\begin{equation}
\frac{dr}{dt}=\left(\derv{t}+{\bf u}{\bf\cdot}\nabla\right)r=-T. \label{eq:4.8}
\end{equation}
Here $r=\beta/\rho$ where $\beta$ is the Clebsch potential that ensures $dS/dt=0$ in the Eulerian, 
Clebsch variational approach (e.g. \cite{Zakharov97}) and 
$d/dt=\partial/\partial t+{\bf u}{\bf\cdot}\nabla$ is the Lagrangian time derivative following the flow. 
\cite{Webb14a},\cite{Webb14b} and \cite{WebbAnco15} show that for fluid dynamics (${\bf B}=0$) the  
modified vorticity flux $\boldsymbol{\Omega}{\bf\cdot}dS$ is advected or Lie dragged with the flow, i.e., 
\begin{equation}
\frac{d}{dt}\left(\boldsymbol{\Omega}{\bf\cdot}dS\right)=
\biggl[\deriv{\boldsymbol{\Omega}}{t}
-\nabla\times\left({\bf u}\times\boldsymbol{\Omega}\right)
+{\bf u}(\nabla{\bf\cdot}\boldsymbol{\Omega})
\biggr]{\bf\cdot}dS=0. \label{eq:4.9}
\end{equation}
The conservation law (\ref{eq:4.9}) and the associated 
conservation law for the modified fluid helicity $\bf{u}{\bf\cdot}\boldsymbol{\Omega}$  
are nonlocal conservation laws that depend on the nonlocal variable $r=-\int_0^t T({\bf x},t)dt$ 
where the integration is with respect to the Lagrangian time $t$ (e.g. \cite{Webb14a, Webb14b}). 
Conservation law (\ref{eq:4.9})
in fluid dynamics is analogous to Faraday's equation in MHD. 
Note that $\nabla{\bf\cdot}\boldsymbol{\Omega}=0$ 
in (\ref{eq:4.9}).  

There is another symplecticity conservation law obtained from (\ref{eq:msymp11}) for the case $\beta=0$ 
and $\gamma=1,2,3$. In that case (\ref{eq:msymp11}) reduces to:
\begin{equation}
\derv{t}{\bf G}-\nabla G_0=0, \label{eq:4.10}
\end{equation}
where ${\bf G}=0$ is the momentum equation (\ref{eq:msymp13}) and $G_0=0$ is the energy conservation 
equation (\ref{eq:msymp12}). 

The general form of the symplecticity equations for MHD using Eulerian Clebsch potentials 
were given in \cite{Webb14c} [equations (5.44) et seq. of that paper]. There were some typographical 
errors in the flux $F^k_{ab}$, indicated below. The general symplecticity conservation laws 
obtained by \cite{Webb14c} have the form:
\begin{equation}
\derv{t}\left(F^0_{ab}\right)+\derv{x^k}\left(F^k_{ab}\right)=0, \label{eq:4.11}
\end{equation}
where
\begin{equation}
F^0_{ab}=\frac{\partial(\phi,\rho)}{\partial(x^a,x^b)}
+ \frac{\partial(\beta,S)}{\partial(x^a,x^b)}
+\frac{\partial(\lambda,\mu)}{\partial(x^a,x^b)}
+\frac{\partial(\Gamma_s,B^s)}{\partial(x^a,x^b)}, \label{eq:4.12}
\end{equation}
and
\begin{align}
F^k_{ab}=&-\frac{\partial(\rho u^k,\phi)}{\partial(x^a,x^b)}+\frac{\partial(\beta u^k,S)}{\partial(x^a,x^b)}
+\frac{\partial(\lambda u^k,\mu)}{\partial(x^a,x^b)}\nonumber\\
&+\frac{\partial(\Gamma_s B^s,u^k)}{\partial(x^a,x^b)}+\frac{\partial(\Gamma_s u^k,B^s)}{\partial(x^a,x^b)}
-\frac{\partial(\Gamma_s B^k,u^s)}{\partial(x^a,x^b)}. \label{eq:4.13}
\end{align}
In (\ref{eq:4.13}) $1\leq k\leq 3$. 
The fourth term on the right handside of (\ref{eq:4.13})
was missed in equation (5.48) in \cite{Webb14c}. Also in (\ref{eq:4.13}) we used the identity:
\begin{equation}
\frac{\partial(\phi u^k,\rho)}{\partial(x^a,x^b)}+\frac{\partial(\rho \phi,u^k)}{\partial(x^a,x^b)}
=-\frac{\partial(\rho u^k,\phi)}{\partial(x^a,x^b)}, \label{eq:4.14}
\end{equation}
to simplify (5.48) of \cite{Webb14c}.  The derivation of the symplecticity 
conservation laws (\ref{eq:4.3}) and (\ref{eq:4.10}) using the general symplecticity laws (\ref{eq:4.11})
and using (\ref{eq:Clebsch3})-(\ref{eq:Clebsch9}) is a non-trivial algebraic exercise.  

\subsection{Differential forms approach}
Proposition 4.3 in \cite{Webb14c} contains flaws. A consistent approach to the multi-symplectic 
equations using differential forms for 1D Lagrangian gas dynamics was given by \cite{Webb15}. 
\cite{WebbAnco15} 
 have given the corresponding theory for multi-dimensional, ideal, compressible,  Lagrangian gas dynamics. 
Below we use  
differential forms to describe the Eulerian, Clebsch variable MHD variational principle of 
\cite{Webb14c}.   

\begin{proposition}
The multi-symplectic system (\ref{eq:msymp1}) is a stationary point of the action:
\begin{equation}
J=\int\psi^* (\Theta)=\int L dV, \label{eq:4.15}
\end{equation}
where $\psi^*(\Theta)$ is the pullback of the differential form $\Theta$ given below, namely:
\begin{align}
\Theta=&\omega^\alpha\wedge d\tilde{x}^\alpha-H dV,\quad \omega^\alpha=L^\alpha_j dz^j, \nonumber\\
dV=&dt\wedge dx\wedge dy\wedge dz,\quad d\tilde{x}^\alpha=\partial_\alpha\lrcorner dV
\equiv (-1)^\alpha dx_0\wedge\ldots dx^{\alpha-1}\wedge dx^{\alpha+1}\ldots \wedge dx^n, \label{eq:4.16}
\end{align}
where we use the notation $(x^0,x^1,x^2,x^3)=(t,x,y,z)$, and $L$ is the constrained Lagrangian 
(\ref{eq:Clebsch2}).
\end{proposition}

\begin{proof}
The pullback of the form $\Theta$ is given by:
\begin{align}
\psi^*(\Theta)=&\psi^*\left(L^\alpha_j dz^j\wedge d\tilde{x}_\alpha-H dV\right)\nonumber\\
=&L^\alpha_j \deriv{z^j}{x^s} dx^s\wedge d\tilde{x}^\alpha-HdV. \label{eq:4.17}
\end{align}
However, 
\begin{equation}
dx^s\wedge d\tilde{x}_\alpha
=dx^s\wedge (-1)^\alpha dx^0\ldots \wedge dx^{\alpha-1}\wedge dx^{\alpha+1} \ldots\wedge dx^n
\equiv (-1)^{2\alpha} \delta^s_\alpha dV. \label{eq:4.18}
\end{equation}
Thus,
\begin{equation}
\psi^*(\Theta)=\left(L^\alpha_j \deriv{z^j}{x^\alpha}-H\right) dV\equiv  L dV, \label{eq:4.19}
\end{equation}
where $L$ is the multi-symplectic Lagrangian (\ref{eq:msymp2}). 

The stationary point conditions, $\delta J/\delta z^i=0$, give the Euler-Lagrange equations:
\begin{equation}
\frac{\delta J}{\delta z^i}=\deriv{L}{z^i}
-\derv{x^j}\left(\deriv{L}{z^i_{,j}}\right)
={\sf K}^\alpha_{ij}
\deriv{z^j}{x^\alpha}-\deriv{H}{z^i}=0, \label{eq:4.20}
\end{equation}
which is the multi-symplectic system (\ref{eq:msymp1}) (see also \cite{Hydon05}).
\end{proof}

\begin{proposition}
Consider the variational functional:
\begin{equation}
G[\Omega]=\int_M \Omega, \label{eq:4.21}
\end{equation} 
where
\begin{equation}
\Omega=d\Theta=d\omega^\alpha\wedge d\tilde{x}_\alpha-dH\wedge dV. \label{eq:4.22}
\end{equation}
and $M$ is a region of the jet space (fiber bundle space) with boundary $\partial M$, in which 
the $z^s$ are taken as independent of the base variables $x^\alpha$ ($\alpha=0,1,2,3$). 
Consider the variational principle:
\begin{equation}
\delta G[\Omega]=\int_M {\cal{L}}_{\bf V}\left(\Omega\right)=0, \label{eq:4.23}
\end{equation}
where
\begin{equation}
{\cal{L}}_{\bf V}=\frac{d}{d\epsilon}=V^i\derv{z^i}, \label{eq:4.24}
\end{equation}
is the Lie derivative with respect to the arbitrary, but smooth vector field ${\bf V}$. 
The variational equation  $\delta G[\Omega]=0$ reduces to:
\begin{equation}
\delta G[\Omega]=\int_{\partial M} V^p\beta_p=0, \label{eq:4.25}
\end{equation}
where the forms $\beta_p$ are given by the formulae:
\begin{equation}
\beta_p=\derv{z^p}\lrcorner\Omega
={\sf K}^\alpha_{pj} dz^j\wedge d\tilde{x}_\alpha-\deriv{H}{z^p} dV. \label{eq:4.26}
\end{equation}
($1\leq p\leq N$). Because the $V^p$ are arbitrary smooth functions of the $z^s$, the variational principle 
$\delta G[\Omega]=0$ implies:
\begin{equation}
\beta_p=0,\quad 1\leq p\leq N. \label{eq:4.27}
\end{equation}
The pullback of the forms $\{\beta_p\}$ to the base manifold gives the equations:
\begin{equation}
\tilde{\beta}_p=\left({\sf K}^\alpha_{pj} \deriv{z^j}{x^\alpha}-\deriv{H}{z^p}\right) dV=0. \label{eq:4.28}
\end{equation}
Thus, the sectioned forms $\tilde{\beta}_p$ vanish on the solution manifold of the multi-symplectic system 
(\ref{eq:msymp1}), and the $\{\beta_p\}$ can be used as a basis of Cartan forms describing 
the system (\ref{eq:msymp1}).
\end{proposition}

\begin{proof}
The proof is essentially the same as that given by \cite{Webb15} for the case of 1D gas dynamics. 
A critical component of the proof is the use of Cartan's magic formula:
\begin{equation}
{\cal{L}}_{\bf V}\left(\Omega\right)={\bf V}\lrcorner d\Omega+d \left({\bf V}\lrcorner \Omega\right)
=d \left({\bf V}\lrcorner \Omega\right), \label{eq:4.29}
\end{equation}
where we used the facts $\Omega=d\Theta$ and $d\Omega=dd\Theta=0$. 
Using (\ref{eq:4.29}) and Stokes theorem, (\ref{eq:4.23})  reduces to: 
\begin{equation}
\delta G[\Omega]=\int_M d({\bf V}\lrcorner\Omega)=\int_{\partial M} {\bf V}\lrcorner\Omega
=\int_{\partial M} V^p\left(\derv{z^p}\lrcorner\Omega\right)=\int_{\partial M} V^p\beta_p=0, 
\label{eq:4.30}
\end{equation}
which verifies (\ref{eq:4.25}). The formula (\ref{eq:4.26}) for $\beta_p$ is obtained by using 
(\ref{eq:msymp5a}) for $d\omega^\alpha$ and  (\ref{eq:4.22}) for $\Omega$, to obtain:
\begin{equation}
\beta_p=\derv{z^p}\lrcorner\Omega=\derv{z^p}\lrcorner\left(\frac{1}{2} 
{\sf K}^\alpha_{ij} dz^i\wedge dz^j\wedge d\tilde{x}_\alpha -\deriv{H}{z^a} dz^a\wedge dV\right), 
\label{eq:4.31}
\end{equation}
Using the skew symmetry of ${\sf K}^\alpha_{ij}$ and $dz^i\wedge dz^j$, (\ref{eq:4.31}) reduces 
to the expression (\ref{eq:4.26}) for $\beta_p$. This completes the proof.
\end{proof}

\subsection{The differential forms $\beta_p$}
The differential forms $\beta_p=\partial_{z^p}\lrcorner\Omega$ in (\ref{eq:4.26}) may be used to 
represent the MHD system described by the Clebsch variable variational principle. The dependent variables
${\bf z}$ are listed in (\ref{eq:msymp2}). In the time evolution variational principle 
(e.g. \cite{Zakharov97}, the fluid velocity ${\bf u}$ is expressed in terms of 
the Clebsch potentials, and is eliminated from the Hamiltonian density 
$H=(1/2)\rho u^2 +\varepsilon(\rho,S)+B^2/(2\mu)$ 
and $(\rho,\phi)$, $(S,\beta)$, $(\mu,\lambda)$, $({\bf B},\boldsymbol{\Gamma})$ are 
canonically conjugate pairs in the canonical Poisson bracket. We use the notation:
\begin{equation}
\beta^{z^i}=\partial_{z^i}\lrcorner\Omega, \label{eq:4.32}
\end{equation}
where the Cartan Poincar\'e form $\Omega$ in (\ref{eq:4.22}) has the form:
\begin{equation}
\Omega=d\omega^0\wedge d\tilde{x}_0+d\omega^k\wedge d\tilde{x}_k-\deriv{H}{z^p}dz^p\wedge dV, \label{eq:4.33}
\end{equation}
where the differential forms $\omega^0$ and $\omega^k$ are listed in (\ref{eq:msymp7}). 
From (\ref{eq:4.32}) and (\ref{eq:4.33}) 
\begin{equation}
\beta^{z^i}=\partial_{z^i}\lrcorner\left\{d\omega^0\wedge d\tilde{x}_0+d\omega^k\wedge d\tilde{x}_k
-\deriv{H}{z^p} dz^p\wedge dV\right\}. \label{eq:4.34}
\end{equation}

Using (\ref{eq:msymp7}) and (\ref{eq:4.34}) we obtain:
\begin{equation}
\beta^{u^i}
=\left(\beta dS+\lambda d\mu-\rho d\phi- B^s d\Gamma_s\right)\wedge d\tilde{x}_i
+\left(\Gamma_id B^k+B^k d\Gamma_i\right)\wedge d\tilde{x}_k+\rho u^i dV, \label{eq:4.35}
\end{equation}
for the differential forms associated with ${\bf u}$. Using the identity
\begin{equation}
dx^a\wedge d\tilde{x}_i=\delta^a_i dV, \label{eq:4.36}
\end{equation}
the sectioned form equation $\tilde{\beta}^{u^i}=0$ yields the expression:
\begin{equation}
\rho {\bf u}=\rho\nabla\phi-\beta\nabla S-\lambda\nabla\mu
+{\bf B}{\bf\cdot}(\nabla\boldsymbol{\Gamma})^T
-{\bf B}{\bf\cdot}\nabla\boldsymbol{\Gamma}-\boldsymbol{\Gamma}\nabla{\bf\cdot}{\bf B}, \label{eq:4.37}
\end{equation}
which is equivalent to the Clebsch expansion for the mass flux $\rho {\bf u}$ given in (\ref{eq:Clebsch3}).

The differential form $\beta^\rho$ is given by:
\begin{equation}
\beta^\rho=\partial_\rho\lrcorner \Omega=(\partial_\rho \lrcorner d\omega^0)\wedge d\tilde{x}_0 
+(\partial_\rho\lrcorner d\omega^k)\wedge d\tilde{x}^k-\deriv{H}{\rho} dV. \label{eq:4.38}
\end{equation}
Using (\ref{eq:msymp5}) for $H$, we obtain:
\begin{equation}
\deriv{H}{\rho}=-\left(\frac{1}{2} u^2-\varepsilon_\rho\right)=-\left(\frac{1}{2} u^2-h\right), 
\label{eq:4.39}
\end{equation}
where $h=(\varepsilon+p)/\rho$ is the gas enthalpy. Substituting (\ref{eq:4.39}) in (\ref{eq:4.38}) 
gives
\begin{align}
\beta^\rho=&-\left(d\phi\wedge d\tilde{x}_0+u^k d\phi\wedge d\tilde{x}_k\right)+\left(\frac{1}{2} u^2-h\right) dV, \label{eq:4.40}\\
\tilde{\beta}^\rho=&-\left[\frac{d\phi}{dt}-\left(\frac{1}{2} u^2-h\right)\right] dV, \label{eq:4.41}
\end{align}
for the differential form $\beta^\rho$ and for the sectioned form $\tilde{\beta}^\rho$. Note that 
$\tilde{\beta}^\rho=0$ is equivalent to Bernoulli's equation (\ref{eq:Clebsch6}). 
Similarly, we obtain:
\begin{align}
\beta^S=&\partial_S\lrcorner \Omega=-d\beta\wedge d\tilde{x}_0
-d(\beta u^k)\wedge d\tilde{x}_k-\rho T dV, \label{eq:4.42}\\
\tilde{\beta}^S=&-\left(\deriv{\beta}{t}+\nabla{\bf\cdot}(\beta {\bf u})+\rho T\right) dV. \label{eq:4.43}
\end{align}
The equation $\tilde{\beta}^S=0$ corresponds to (\ref{eq:Clebsch7}) for $\beta$.

Following the above procedure we obtain the equations:
\begin{align}
\beta^\mu=&\partial_\mu\lrcorner \Omega=
-\left[d\lambda\wedge d\tilde{x}_0+d(\lambda u^k)\wedge d\tilde{x}_k\right], \nonumber\\
\beta^{B^i}=&-\left[d\Gamma_i\wedge d\tilde{x}_0+u^k d\Gamma_i\wedge d\tilde{x}_k 
+\Gamma_s du^s\wedge d\tilde{x}_i +\frac{B^i}{\mu}dV \right], \nonumber\\
\beta^{\Gamma_i}=&dB^i\wedge d\tilde{x}_0+ \left(B^i du^k+ u^k dB^i-B^k du^i\right)\wedge d\tilde{x}_k, 
\nonumber\\
\beta^\lambda=&\partial_\lambda\lrcorner\Omega
=\left(d\mu\wedge d\tilde{x}_0+u^k d\mu\wedge d\tilde{x}_k\right), \nonumber\\
\beta^\beta=&\partial_\beta\lrcorner\Omega=dS\wedge d\tilde{x}_0+u^k dS \wedge d\tilde{x}_k,
\nonumber\\
\beta^\phi=&\partial_\phi\lrcorner \Omega
=d\rho\wedge d\tilde{x}_0
+\left(u^k d\rho+\rho du^k\right)\wedge d\tilde{x}_k.  \label{eq:4.44}
\end{align}
The pullback of the above equations, i.e. $\tilde{\beta}^{z^p}=0$, gives the evolution equations 
for $(\lambda,\boldsymbol{\Gamma}, {\bf B}, \mu, S, \rho)$ listed in (\ref{eq:Clebsch5})-(\ref{eq:Clebsch9}). 
The differential form equations $\beta^{z^p}=0$ thus represent the partial differential 
equation system (\ref{eq:msymp1}). 

It is not obvious that the system of forms $\{\beta^{z^p}\}$ above is a closed ideal. A  
check on the closure of the forms for the case of non-barotropic, 
1D gas dynamics (i.e. ${\bf B}=0$) indicates that the 
ideal of forms 
${\cal{I}}=\{\beta^u,\beta^\rho,\beta^S,\beta^\beta,\beta^\phi\}$  
can be closed by adjoining 
the form $d\beta^u$. The ideal $\cal{I}$ is closed for the case of a barotropic gas.  
 The Cartan approach to Lie symmetries 
requires that ${\cal I}$ is a closed ideal (e.g. Harrison and Estabrook (1971)). 
The ideal is closed  if $d\beta_i=c_{ij}\wedge \beta^{z^j}$, where the $c_{ij}$ are forms. 
It should be noted that the ideal of forms obtained by \cite{Webb15} for 1D, Lagrangian, multi-symplectic 
gas dynamics is closed. The ideal of forms for multi-dimensional, Lagrangian. compressible gas dynamics
obtained from the Cartan-Poincar\'e form is also a closed ideal (Webb and Anco (2015)). 
This suggests that the ideal of forms using the Clebsch variable description has a more complicated 
structure than the set of forms that arise in the Lagrangian variational approach.
In general,  this question requires more work 
and lies beyond the scope of the present paper. 

\section{Summary and Concluding Remarks}
In this paper, we have corrected some flaws in the paper by \cite{Webb14c} on multi-symplectic 
MHD. This includes the correction of some typographic errors in the expressions for the 
conserved fluxes in the symplecticity conservation laws written in terms of the Clebsch 
potentials in (\ref{eq:4.11})-(\ref{eq:4.13}) and in particular expression (\ref{eq:4.13}) 
for $F^k_{ab}$ has been corrected. We pointed out that the vorticity-type symplecticity 
conservation law is equivalent to taking the curl of the MHD momentum equation in the form 
(\ref{eq:4.6}).  This conservation law has conserved density $\nabla\times{\bf M}$
where ${\bf M}=\rho {\bf u}$ is the momentum density of the MHD fluid. Pressure gradient 
forces due to the combined gas and magnetic pressures (i.e. $\nabla(p+p_B)$, where $p_B=B^2/(2\mu)$ 
is the magnetic pressure and $p$ is the gas pressure), play no role in this conservation law. 
The conservation law (\ref{eq:4.6}) describes fluid motions with spin and rotation 
(e.g; as in an Alfv\'en wave). This symplecticity conservation law is different than that obtained 
by \cite{WebbAnco15} in ideal Lagrangian fluid mechanics. 
In the latter paper the vorticity 
symplecticity conservation law involved 
$\boldsymbol{\Omega}=\boldsymbol{\omega}+\nabla r\times\nabla S$ 
where $\boldsymbol{\omega}=\nabla\times{\bf u}$ is the fluid vorticity and $r=\beta/\rho$ where 
$\beta$ is the Clebsch potential that enforces entropy conservation following the flow. This 
 is a nonlocal conservation law as it involves the nonlocal variable $r$. 
It has the same form as Faraday's equation in MHD except that 
${\bf B}$ is replaced by $\boldsymbol{\Omega}$.  
A further symplecticity conservation law is due to a compatibility condition 
which requires that time derivative of the 
 the momentum conservation equation minus the gradient of the energy conservation equation 
is zero. We provided a more consistent treatment of variational principles 
associated with the Cartan-Poincar\'e form than that given in \cite{Webb14c}, and its relation 
to Cartan's geometric theory of partial differential equations using differential forms (Sections 
4.1 and 4.2). \cite{Cendra13} suggest that the formulation of multi-symplectic 
systems using exterior differential forms, is a natural approach in 
taking into account integrability conditions in the theory. 

\section*{Acknowledgements}
GMW acknowledges discussions with Darryl Holm and Phil. Morrison on multi-symplectic MHD.
GMW is supported in part by NASA grant NNX15A165G. GPZ was supported
in part by NASA grants NN05GG83G and NSF grant nos. ATM-03-17509 and ATM-04-28880.

\end{document}